\documentclass[letterpaper, 10 pt, conference]{ieeeconf}  

\IEEEoverridecommandlockouts                              
\usepackage{cite}
\usepackage{amsmath,amssymb,amsfonts}

\usepackage{tikz}
\usetikzlibrary{arrows.meta,positioning,fit,calc}

\usepackage{graphicx}
\usepackage{textcomp}
\usepackage{xcolor}

\usepackage{physics}

\usepackage{algorithm}
\usepackage{algpseudocode}
\def\BibTeX{{\rm B\kern-.05em{\sc i\kern-.025em b}\kern-.08em
    T\kern-.1667em\lower.7ex\hbox{E}\kern-.125emX}}

\usepackage[top=0.8in, bottom=0.85in, left=0.77in, right=0.77in]{geometry}

\usepackage{tikz}
\usetikzlibrary{positioning,fit,arrows.meta,calc}

\usepackage{subcaption} 

\newtheorem{assumption}{Assumption}

\newtheorem{proposition}{Proposition}
\newtheorem{theorem}{Theorem}
\newtheorem{lemma}{Lemma}

\newtheorem{remark}{Remark}

\usepackage{mathrsfs}

\title{\LARGE \bf
Tractable Infinite-Horizon Stochastic Model Predictive Control for Quantum Filtering via Eigenstate Reduction}

\author{Yunyan Lee, Ian R. Petersen, \IEEEmembership{Fellow, IEEE}, and Daoyi Dong, \IEEEmembership{Fellow, IEEE}
\thanks{Yunyan Lee and Ian R. Petersen are with the School of Engineering, The 
Australian National University, Canberra, ACT 2601, Australia. (email: {\tt\small Yun-Yan.Lee@anu.edu.au, ian.petersen@anu.edu.au})
} 
\thanks{Daoyi Dong is with the Australian Artificial Intelligence Institute, Faculty of Engineering and Information Technology, University of Technology Sydney, NSW 2007, Australia. (email: {\tt\small daoyidong@gmail.com})
}
}

\begin{document}
\maketitle
\thispagestyle{empty}
\pagestyle{empty}

\begin{abstract}
Model predictive control has shown potential to enhance the robustness of quantum control systems. In this work, we  propose a tractable Stochastic Model Predictive Control (SMPC) framework for finite-dimensional quantum systems under continuous-time measurement and quantum filtering. Using the almost-sure eigenstate reduction of quantum trajectories, we prove that the infinite-horizon stochastic objective collapses to a fidelity term that is computable in closed form from the one-step averaged state. Consequently, the online SMPC step requires only deterministic propagation of the filter and a terminal fidelity evaluation. An advantage of this method is that it eliminates per-horizon Monte Carlo scenario sampling and significantly reduces computational load while retaining the essential stochastic dynamics. We establish equivalence and mean-square stability guarantees, and validate the approach on multi-level and Ising-type systems, demonstrating favorable scalability compared to sampling-based SMPC.
\end{abstract}


\section{INTRODUCTION}
\label{sec:introduction}
Quantum feedback control strategies are typically classified into coherent quantum feedback and measurement-based feedback \cite{Zhang2017Quantum}. Coherent control utilizes quantum controllers interacting directly with the quantum plant \cite{James2008H, Maalouf2009Coherent, Nurdin2009Coherent}, with implementations in optical systems \cite{Mabuchi2008Coherent}, ion traps \cite{Lloyd2000Coherent}, and superconducting circuits \cite{vijay2012stabilizing}. In contrast, measurement-based feedback uses classical controllers\cite{Belavkin1992Quantum, Bouten2007An, Doherty1999Feedback, Mirrahimi2007Stabilizing, Handel2005Feedback, Wiseman1993Quantum}, and has been applied to tasks such as quantum error correction \cite{Chase2008Efficient}, entanglement stabilization \cite{Carvalho2008Controlling, Liu2016Comparing}, and control of atomic and BEC systems \cite{Campagne2016Using, Szigeti2009Continuous}.

In this work, we focus on quantum measurement-based feedback implemented via quantum filtering~\cite{belavkin1987non,Belavkin1992Quantum}. As a canonical physical realization of quantum filtering, the Laboratoire Kastler–Brossel photon-box experiment stabilizes cavity Fock states via real-time estimation and feedback, where probe atoms acquire phase information correlated with the photon number and the resulting measurement record is processed online to drive the cavity toward a target state~\cite{sayrin2011real,amini2013feedback}; see Fig.~\ref{fig:min_feedback_simple}.

Various approaches, such as Lyapunov-based control and stochastic optimal control methods, have been developed for quantum filtering problems~\cite{dong2010quantum,dong2023learning,wiseman2010quantum}. Several studies have further advanced the understanding of quantum filtering systems, including stability and convergence of continuous-time quantum filters~\cite{amini2014stability,song2016continuous}, robust feedback stabilization for multi-level systems~\cite{liang2021robust,liang2025exploring}, and estimation techniques for systems with unknown initial states~\cite{liang2020estimation}. A comprehensive tutorial on these topics can be found in~\cite{rouchon2022tutorial}.

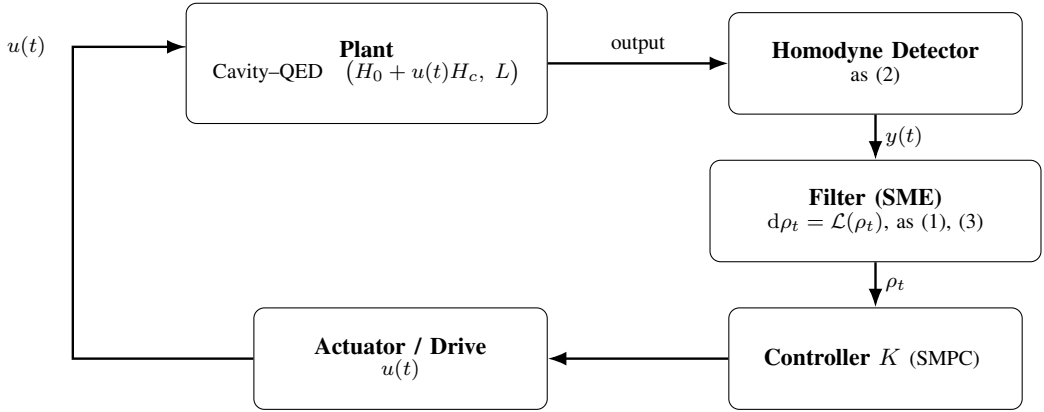
\begin{figure*}[t]
\centering
\begin{tikzpicture}[
  font=\small, >=Latex,
  box/.style={draw,rounded corners,minimum width=3.9cm,minimum height=1.35cm,align=center},
  boxL/.style={draw,rounded corners,minimum width=4.8cm,minimum height=1.6cm,align=center},
  arrow/.style={-Latex, line width=0.9pt}
]

\node[boxL] (Plant) {%
  \textbf{Plant}\\[-2pt]
  \footnotesize Cavity--QED\quad$\big(H_0+u(t)H_c,\;L\big)$
};

\node[box,right=2.4cm of Plant] (HD) {%
  \textbf{Homodyne Detector}\\[-1pt]
  \footnotesize as~\eqref{eqn:dy}
};

\node[box,below=0.6cm of HD,minimum width=4.4cm] (SME) {%
  \textbf{Filter (SME)}\\[-1pt]
  \footnotesize $\mathrm{d}\rho_t=\mathcal{L}(\rho_t)$, as~\eqref{eqn:SME},~\eqref{eqn:innovations}
};

\node[box,below=0.6cm of SME] (K) {%
  \textbf{Controller} $K$ \footnotesize(SMPC)
};

\node[box,left=2.4cm of K] (Act) {%
  \textbf{Actuator / Drive}\\[-2pt]\footnotesize $u(t)$
};

\draw[arrow] (Plant) -- node[above]{\footnotesize output} (HD);
\draw[arrow] (HD) -- node[right]{\footnotesize $y(t)$} (SME);
\draw[arrow] (SME) -- node[right]{\footnotesize $\rho_t$} (K);
\draw[arrow] (K) -- (Act);

\draw ($(Plant.west)+(0,0.22)$) -- ++(-0.18,0);

\draw[arrow]
  (Act.west) -| ($(Plant.west)+(-1.5,0.22)$)
  -- node[left,pos=-0.15]{\footnotesize $u(t)$} ($(Plant.west)+(0,0.22)$);

\end{tikzpicture}
\caption{Block diagram for measurement-based feedback. The homodyne output $y(t)$ is filtered to obtain $\rho_t$, which the controller uses to compute the drive $u(t)$ applied to the plant; see \eqref{eqn:SME}, \eqref{eqn:dy}, \eqref{eqn:innovations}. This corresponds to the measurement–feedback setting in \cite{sayrin2011real,amini2014stability,Bouten2007An}.}
\label{fig:min_feedback_simple}
\end{figure*}

Building on these foundations, we present a model predictive control (MPC) approach tailored to quantum filtering dynamics. MPC has become a key method for controlling systems with constraints and uncertainties, and it has been widely applied across various domains~\cite{grune2017nonlinear, schwenzer2021review}, including quantum systems~\cite{clouatre2022model, goldschmidt2022model, hashimoto2017stability, lee2024robust}. Motivated by the extensive literature on MPC for classical stochastic systems~\cite{Lorenzen2019Stochastic, McAllister2021Stochastic, Lorenzen2017Constraint, Cannon2011Stochastic, Carpintero2021Convergence}, we investigate its applicability to quantum filtering dynamics.

Stochastic MPC (SMPC) offers distinct advantages in quantum settings: it naturally handles physical constraints inherent to quantum hardware, systematically incorporates uncertainties from quantum noise and imperfect measurements, and enables predictive planning over a receding horizon to improve robustness and stability. However, the application of SMPC to general multi-level quantum systems encounters two major challenges: the curse of dimensionality due to exponential state-space growth, and the high computational burden of Monte Carlo sampling required to model measurement-induced randomness.

To overcome these difficulties, we propose a formulation of infinite-horizon SMPC tailored to quantum filtering dynamics. By using quantum reduction \cite{Adler2001Martingale, liang2019exponential}, namely, their almost-sure collapse to eigenstates of the measurement operator, we show that the infinite-horizon cost can be analytically reduced to a finite sum involving eigenstate probabilities. This leads to an equivalent optimization problem with a fidelity-based objective that preserves long-term control goals while significantly reducing computational complexity. The result is tractable for scalable SMPC on multi-level open quantum systems.

The organization of this paper is as follows. Section~\ref{sec:preliminary} outlines the quantum filtering equation and introduces the SMPC control problem. Section~\ref{sec:main_thm} presents our main result, reformulating the SMPC problem into a tractable form. Section~\ref{sec:stability} provides a detailed stability analysis of the resulting closed-loop system. Section~\ref{sec:PMP} describes a candidate control synthesis method based on Pontryagin's Maximum Principle (PMP). The effectiveness of the proposed approach is demonstrated through numerical simulations in Section~\ref{sec:numerical}.

\section{System Description and Problem Statement}
\label{sec:preliminary}
We consider a finite-dimensional open quantum system, described by a density matrix
$\rho(t) \in \mathbb{C}^{N \times N}$. The system evolves under the influence of both control
and continuous-time measurement. The evolution of $\rho(t)$ is governed by the stochastic master equation (SME) \cite{Bouten2007An}:
\begin{align}
\label{eqn:SME}
\begin{aligned}
\mathrm{d}\rho_t &= -i\,[H_0 + H_u(t),\, \rho_t]\,\mathrm{d}t
+ L \rho_t L^\dagger \,\mathrm{d}t - \tfrac{1}{2} \{L^\dagger L, \rho_t\} \,\mathrm{d}t \\
&\quad + \sqrt{\eta}\,\Big( L \rho_t + \rho_t L^\dagger
- \Tr[(L + L^\dagger)\rho_t]\,\rho_t \Big)\,\mathrm{d}W_t,
\end{aligned}
\end{align}
where $H_0$ is the Hamiltonian, $H_u(t)$ is the control Hamiltonian,
$L$ is the measurement operator, $\eta\in[0,1]$ is the measurement efficiency,
$\mathrm{d}W_t$ is a standard Wiener increment, $i=\sqrt{-1}$, and
$\Tr(\cdot)$ denotes the trace; we write $[A,B]=AB-BA$, $\{A,B\}=AB+BA$, and assume $\hbar=1$.

The observed output process corresponding to the homodyne photocurrent $Y_t$ satisfies
\begin{align}
\label{eqn:dy}
\mathrm{d}Y_t
= \sqrt{\eta}\,\Tr\!\big[(L{+}L^\dagger)\rho_t\big]\,\mathrm{d}t \;+\; \mathrm{d}W_t,
\end{align}
so that the innovations process used in \eqref{eqn:SME} is
\begin{align}
\label{eqn:innovations}
\mathrm{d}W_t
:= \mathrm{d}Y_t - \sqrt{\eta}\,\Tr\!\big[(L{+}L^\dagger)\rho_t\big]\,\mathrm{d}t,
\end{align}
which is a standard Wiener increment with respect to the filtration generated by $\{Y_s: s\le t\}$.
Equivalently, the instantaneous photocurrent $y(t)$ is
\begin{align}
\label{eqn:output}
y(t) \;=\; \sqrt{\eta}\,\Tr\!\big[(L{+}L^\dagger)\rho_t\big] \;+\; \xi(t),
\quad \text{with } \xi(t) := \tfrac{\mathrm{d}W_t}{\mathrm{d}t},
\end{align}
interpreting $\xi(t)$ as unit spectral-density white noise.

We begin with the configuration \( H_0 = L = J_z \), where \( J_z \) denotes the angular momentum operator along the \( z \)-axis, a standard choice in quantum spin and optical systems. In the absence of control, the SME dynamics drive the state \( \rho(t) \) almost surely to one of the eigenstates \( \tilde{\rho}_j = |\phi_j\rangle\langle \phi_j| \) with \( J_z |\phi_j\rangle = \lambda_j |\phi_j\rangle \)~\cite{Adler2001Martingale,liang2019exponential}. The control objective is to steer the system toward a designated eigenstate \( \rho_f = \tilde{\rho}_j \) via feedback based on the filtered state. A natural approach is to discretize the SME~\eqref{eqn:SME} and apply an SMPC framework developed for classical stochastic systems~\cite{Lorenzen2019Stochastic}. After considering this case, we relax the assumption and extend the analysis to more general configurations where \( H_0 \) and \( L \) need not coincide but share a common eigenbasis.

As discussed in~\cite{Lorenzen2019Stochastic}, applying SMPC to stochastic control of multi-level systems is hindered by the rapid growth of the scenario tree and the associated sample complexity, as detailed in Appendix. This challenge is exacerbated in quantum systems, where the state-space dimension grows exponentially with the number of qubits, rendering direct SMPC for quantum filtering dynamics computationally infeasible even for moderate system sizes.

To address this challenge, we introduce an infinite-horizon cost functional that captures the long-term behavior of the SME dynamics.
\begin{align}
    \label{eqn:SMPC-cost}
    J(u) := \lim_{T \to \infty}\frac{1}{T}\int_{\Omega}\int_{t}^{T} d_B^2\bigl(\rho(s,\omega,u), \rho_f\bigr)\, ds\, d\mu(\omega),
\end{align}
where \( \omega \in \Omega \) denotes the realization of measurement noise, \( \mu \) is the associated probability measure, and \( d_B(\rho_a, \rho_b) \) is the Bures distance:
\begin{align}
    d_B(\rho_a, \rho_b) := \sqrt{2 - 2\, \operatorname{Tr} \left( \sqrt{ \sqrt{\rho_b} \rho_a \sqrt{\rho_b} } \right)}.
\end{align}
In the special case where \( \rho_b = \tilde{\rho}_j = |\phi_j\rangle\langle \phi_j| \) is a pure state, this simplifies to
\begin{align}
    d_B(\rho_a, \tilde{\rho}_j) = \sqrt{2 - 2\, \operatorname{Tr}(\rho_a \tilde{\rho}_j)}.
\end{align}

Although the inclusion of an infinite-horizon objective in a stochastic, multi-level quantum setting may initially appear intractable, we show that the structure of quantum trajectory collapse enables a reduction of the cost in \eqref{eqn:SMPC-cost} to a form that is both computationally tractable and physically interpretable. This result provides a foundation for constructing an efficient SMPC scheme suitable for multi-level quantum systems.

We then formulate an SMPC controller that optimizes this cost subject to the system dynamics and input constraints. The optimization problem is defined as:
\begin{align}
\label{eqn: SMPC problem}
\begin{aligned}
    \min_{u(s)}\quad & J(u)  \text{ given in \eqref{eqn:SMPC-cost}}\\
    \text{s.t.}\quad & 
    \begin{cases}
    \rho(s)\big|_{s = t} = \rho(t),\\
        \rho(s) \text{ evolves under the SME given in \eqref{eqn:SME}}, \\
        u(s) \in \mathcal{U}, \quad \forall s \in [t, t + \Delta t],
    \end{cases}
\end{aligned}
\end{align}
where $\mathcal{U}$ denotes the set of admissible control inputs, and the optimization is performed over a short control horizon $[t, t+\Delta t]$, consistent with a receding horizon framework.

The SMPC controller \eqref{eqn: SMPC problem} captures the long-term control objective through an infinite-horizon cost while requiring only short-term optimization. In the following section, we use a property of quantum state reduction to derive an equivalent, tractable cost function, as formally stated in Theorem~\ref{thm:SMPC_equivalence}.

\section{Problem Reformulation via Eigenstate Decomposition}
\label{sec:main_thm}
In this section, we demonstrate that the infinite-horizon cost
defined in~\eqref{eqn:SMPC-cost} for a stochastic evolution
can be equivalently formulated as a deterministic cost function
based directly on the state fidelity.
To illustrate this connection, we first consider a simple yet nontrivial
example, the angular-momentum system with
$H_0=L=J_z$.
We then discuss how this result extends to a broader class of
measurement designs through appropriate choices of the operator~$L$.

\subsection{Angular-Momentum System}
\label{subsec:angular-system}
To make the infinite-horizon cost \eqref{eqn:SMPC-cost} analytically tractable, we reformulate it by using the structure of quantum state reduction, as characterized in Lemma~\ref{lem:quantum_state_reduction}. It has been shown in~\cite{liang2019exponential} that the solution of the SME~\eqref{eqn:SME} converges almost surely to one of the eigenstates of the measurement operator \( L \). Let \( \{\tilde{\rho}_j\}_{j=1}^N \) denote the set of these eigenstates, which form an orthonormal basis and satisfy \( \sum_j \tilde{\rho}_j = I \).

\begin{lemma}[Quantum State Reduction~{\cite{liang2019exponential}}]
\label{lem:quantum_state_reduction}
Consider the stochastic master equation~\eqref{eqn:SME} with \( u(s) = 0 \). Then, the set of eigenstates \( \{ \tilde{\rho}_j \}_{j=1}^N \) of the measurement operator \( L \) is exponentially stable both in expectation and almost surely. That is, the quantum state \( \rho(t) \) converges to one of the eigenstates \( \tilde{\rho}_j \), with average and sample Lyapunov exponents less than or equal to \( -\eta M/2 \), where \( \eta \) is the measurement efficiency and \( M > 0 \) is a constant determined by the spectrum of \( L \).

Moreover, the probability that \( \rho(t) \) converges to \( \tilde{\rho}_j \) is given by
\[
\mathbb{P}\left(\lim_{t \to \infty} \rho(t) = \tilde{\rho}_j \right) = \operatorname{Tr}[\rho(0)\tilde{\rho}_j], \quad \text{for } j = 1,\dots, N.
\]
\end{lemma}

This property forms the foundation for the following equivalence result, which simplifies the original infinite-horizon SMPC problem \eqref{eqn: SMPC problem} into a tractable optimization problem.

\begin{theorem}[Equivalence of Infinite-Horizon SMPC]
\label{thm:SMPC_equivalence}
Consider the infinite-horizon SMPC problem defined in \eqref{eqn: SMPC problem}, with cost function given by \eqref{eqn:SMPC-cost}, and suppose the target state $\rho_f$ is an eigenstate $\tilde{\rho}_j$ of the measurement operator $L$. Then, the SMPC problem is equivalent to the following optimization problem:
\begin{align}
\label{eq:SMPC-final}
    \min_{u(s),\, s \in [t, t+\Delta t]} 2\left(1 - \operatorname{Tr}\left[\rho(t+\Delta t; u)\rho_f\right]\right),
\end{align}
subject to the system dynamics in \eqref{eqn:SME} and admissible control inputs $u(s) \in \mathcal{U}$.
\end{theorem}
\begin{proof}
We begin by decomposing the infinite-horizon cost functional~\eqref{eqn:SMPC-cost} into two parts:
\begin{align}
    J(u) 
    &= \lim_{T\to\infty}\frac{1}{T}\int_{\Omega}
       \int_{t}^{t+\Delta t} d_B^2(\rho(s,\omega,u),\rho_f)\,ds\,d\mu(\omega) \nonumber\\
    &\quad + \lim_{T\to\infty}\frac{1}{T}\int_{\Omega}
       \int_{t+\Delta t}^{T} d_B^2(\rho(s,\omega,u),\rho_f)\,ds\,d\mu(\omega). \label{eq:cost-split}
\end{align}

The first term corresponds to a finite-time integral over a fixed interval \([t, t+\Delta t]\), which is independent of \(T\). Since \(d_B^2(\cdot,\cdot) \leq 2\), this term is upper-bounded by \(2\Delta t/T\), which vanishes as \(T \to \infty\):
\begin{align}
    \lim_{T \to \infty} \frac{1}{T} \int_{\Omega}
    \int_{t}^{t+\Delta t} d_B^2(\rho(s,\omega,u),\rho_f)\,ds\,d\mu(\omega) = 0.
\end{align}

We now focus on the second term in \eqref{eq:cost-split}. Under the quantum filtering equation~\eqref{eqn:SME}, it is known that each quantum trajectory \(\rho(s,\omega,u)\) converges almost surely to one of the eigenstates \(\{\tilde\rho_j\}_{j=1}^N\) as \(s \to \infty\)~\cite{liang2018exponential}. Therefore, for almost every \(\omega\), we have
\begin{align}
    \lim_{s \to \infty} \rho(s,\omega,u) = \tilde\rho_j, \quad \text{for some } j.
\end{align}

Since \(d_B^2(\rho(s,\omega,u), \rho_f) \to d_B^2(\tilde\rho_j, \rho_f)\) almost surely and is uniformly bounded, we can apply the dominated convergence theorem to exchange the limit and the expectation:
\begin{align}
\label{eq:asymptotic-cost}
\begin{aligned}
      &\lim_{T\to\infty} \frac{1}{T} \int_{\Omega}
    \int_{t+\Delta t}^{T} d_B^2(\rho(s,\omega,u),\rho_f)\,ds\,d\mu(\omega)\\
    &= \sum_{j=1}^N d_B^2(\tilde\rho_j, \rho_f)\, \mathbb{P}_u[\rho(\infty) = \tilde\rho_j].
\end{aligned}
\end{align}

To evaluate the probability \(\mathbb{P}_u[\rho(\infty) = \tilde\rho_j]\), we use the fact that the SME defines a Markov process and that the collapse statistics obey the Born rule. Hence,
\begin{align}
\label{eq:collapse-prob}
\begin{aligned}
      \mathbb{P}_u[\rho(\infty) = \tilde\rho_j] 
    &= \mathbb{E}_\omega\left[ \mathbf{1}_{\{\rho(\infty,\omega;u) = \tilde\rho_j\}} \right] \\
    &= \mathbb{E}_\omega\left[ \operatorname{Tr}(\rho(t+\Delta t, \omega; u)\, \tilde\rho_j) \right] \\
    &= \operatorname{Tr}\left( \mathbb{E}_\omega[\rho(t+\Delta t, \omega; u)] \tilde\rho_j \right) \\
    &= \operatorname{Tr}\left( \rho(t+\Delta t; u) \tilde\rho_j \right),
\end{aligned}
\end{align}
where \( \mathbf{1}_{\{\rho(\infty,\omega;u) = \tilde\rho_j\}} \) is the indicator function that takes the value 1 when the trajectory collapses to \( \tilde\rho_j \), and 0 otherwise; the second equality follows from Lemma~\ref{lem:quantum_state_reduction}. Here, we define the averaged (deterministic) state after evolution as
\begin{align}
\label{eqn:avg_def}
    \rho(t+\Delta t; u) := \mathbb{E}_\omega[\rho(t+\Delta t, \omega; u)].
\end{align}

Substituting \eqref{eq:collapse-prob} into \eqref{eq:asymptotic-cost}, we obtain
\begin{align}
    J(u) = \sum_{j=1}^N d_B^2(\tilde\rho_j, \rho_f) \operatorname{Tr}(\rho(t+\Delta t; u)\, \tilde\rho_j). \label{eq:exact-cost}
\end{align}

Now, if \(\bar\rho_f\) denotes the target eigenstate, then by definition
\begin{align}
    d_B^2(\tilde\rho_j, \bar\rho_f) 
    = \begin{cases}
        0, & \text{if } j = f, \\
        2, & \text{if } j \neq f,
    \end{cases}
\end{align}
and the cost function \eqref{eq:exact-cost} can be simplified to
\begin{align}
\label{eq:final-SMPC-cost}
\begin{aligned}
      J(u) &= \sum_{j \ne f} 2 \cdot \operatorname{Tr}(\rho(t+\Delta t; u)\, \tilde\rho_j) \\
         &= 2 \left(1 - \operatorname{Tr}(\rho(t+\Delta t; u)\, \bar\rho_f) \right),
\end{aligned}
\end{align}
where we have used the relation \( \sum_j \tilde\rho_j = I \).
\end{proof}

As shown in Theorem~\ref{thm:SMPC_equivalence}, the infinite-horizon
SMPC cost can be equivalently expressed as a deterministic
fidelity-based cost function, enabling direct controller design.
We now extend this equivalence beyond the special case
$H_0=L=J_z$ to more general configurations.
In particular, we identify the conditions under which a suitable
target state and measurement operator~$L$
lead to an equivalent deterministic formulation.

\subsection{General Case}
\label{subsec:general}
The proof idea of Lemma~\ref{lem:inv-subspace} follows~\cite{cardona2020exponential};
we adapt their argument to extend Lemma~\ref{lem:quantum_state_reduction} 
to general quantum systems satisfying Assumption~\ref{ass:inv-subspaces}.

\begin{assumption}[Invariant subspaces]\label{ass:inv-subspaces}
There exists a family of orthogonal projectors
$\{\Pi_\alpha\}_{\alpha=1}^m$ such that
\begin{align}
[H_0,\Pi_\alpha]=0,\qquad [L,\Pi_\alpha]=0,\qquad
\sum_{\alpha=1}^m \Pi_\alpha = I .
\end{align}
\end{assumption}

\begin{lemma}
\label{lem:inv-subspace}
Suppose Assumption~\ref{ass:inv-subspaces} holds. Let
$p_\alpha(t):=\Tr(\rho_t\Pi_\alpha)$ denote the population of subspace
$\Pi_\alpha$ under the homodyne SME \eqref{eqn:SME} with efficiency
$\eta\in(0,1]$. Since $[L,\Pi_\alpha]=0$, $L$ and $\{\Pi_\alpha\}$
are simultaneously diagonalizable: if
$L=\sum_{j=1}^d \lambda_j\,\tilde\rho_j$ is a spectral decomposition,
then each $\Pi_\alpha$ is a sum of some $\tilde\rho_j$'s. For any pair
$(\alpha,\beta)$, define a gap
\begin{align}
\label{eqn:gap}
\Delta_{\alpha\beta}
:= \min\big\{|\lambda_j-\lambda_k|:\ 
\tilde\rho_j\le\Pi_\alpha,\ \tilde\rho_k\le\Pi_\beta\big\}.
\end{align}
Then:
\begin{enumerate}
  \item[(i)] $p_\alpha(t)$ is a bounded martingale for each $\alpha$, and
  $\sum_\alpha p_\alpha(t)=1$.

  \item[(ii)] The Lyapunov function 
  \begin{align}
      V(\rho):=\sum_{\alpha<\beta}\sqrt{p_\alpha p_\beta}
  \end{align}
  satisfies 
  \begin{align}
       \mathbb{E}\!\big[V(\rho_t)\big]\ \le\ e^{-rt}\,V(\rho_0),
  \qquad 
  r:=\tfrac{\eta}{2}\min_{\alpha\neq\beta}\Delta_{\alpha\beta}^2.
  \end{align}
  \item[(iii)] $\rho_t$ converges almost surely to the invariant set
  \begin{align}
      \{\rho:\ \exists\alpha,\ p_\alpha(\rho)=1\}.
  \end{align}
  \item[(iv)] For each $\alpha$,
  \begin{align}
        \mathbb{P}\!\big(\rho(\infty)\in\Pi_\alpha\big)
  =\mathbb{E}\!\big[p_\alpha(\infty)\big]
  =\Tr\!\big(\rho(0)\Pi_\alpha\big).
  \end{align}
\end{enumerate}

\end{lemma}

\begin{proof}
Consider a filtered probability space $(\Omega,\mathcal{F},\{\mathcal{F}_t\},\mathbb{P})$
supporting a standard Wiener process $W_t$ adapted to $\{\mathcal{F}_t\}$.
Here $\mathcal{F}_t$ denotes the observation filtration up to time $t$.

(i) For each $\alpha$, set $p_\alpha(t):=\Tr(\rho_t\Pi_\alpha)$ and
$\langle L\rangle_t:=\Tr(L\rho_t)$. Since $[H_0,\Pi_\alpha]=[L,\Pi_\alpha]=0$,
the SME \eqref{eqn:SME} gives
\begin{align}
    \begin{aligned}
        dp_\alpha(t)&=\Tr\!\big(\Pi_\alpha\,d\rho_t\big)\\
&=2\sqrt{\eta}\,\Big(\Tr(\Pi_\alpha L\rho_t)-\langle L\rangle_t\,p_\alpha(t)\Big)\,dW_t,
    \end{aligned}
\end{align}
with no $dt$ term. Hence $p_\alpha(t)\in[0,1]$ and, for any $\Delta t>0$,
\begin{align}
    \mathbb{E}\big[p_\alpha(t+\Delta t)\,\big|\,\mathcal{F}_t\big]=p_\alpha(t),
\end{align}
so $p_\alpha(t)$ is a martingale. Moreover,
\begin{align}
    \sum_\alpha p_\alpha(t)=\sum_\alpha\Tr(\rho_t\Pi_\alpha)
=\Tr\!\Big(\rho_t\sum_\alpha\Pi_\alpha\Big)=1,
\end{align}
which proves (i). By Doob’s martingale convergence theorem, $p_\alpha(t)\to p_\alpha(\infty)$ almost surely as $t\to\infty$.

(ii) Set $\xi_\alpha:=\sqrt{p_\alpha}$ so that 
\begin{align}
V(\rho)=\sum_{\alpha<\beta}\xi_\alpha\xi_\beta .
\end{align}
Because $[L,\Pi_\alpha]=0$, write 
\begin{align}
L=\sum_j \lambda_j \tilde\rho_j, 
\qquad 
\Pi_\alpha=\sum_{j\in J_\alpha}\tilde\rho_j .
\end{align}
Define the group–averaged eigenvalue
\begin{align}
\mu_\alpha(t):=\frac{1}{p_\alpha(t)}\sum_{j\in J_\alpha}\lambda_j\, 
\Tr(\rho_t\tilde\rho_j), 
\\
\varpi(\xi):=\sum_\gamma \mu_\gamma(t)\,\xi_\gamma^2=\Tr(L\rho_t).
\end{align}
From the SME and $[L,\Pi_\alpha]=0$ we have 
\begin{align}
dp_\alpha=2\sqrt{\eta}\,(\mu_\alpha-\varpi)\,p_\alpha\,dW,
\end{align}
hence by using the It\^o rule,
\begin{align}
d\xi_\alpha
&=\frac{1}{2\sqrt{p_\alpha}}\,dp_\alpha
-\frac{1}{8\,p_\alpha^{3/2}}(dp_\alpha)^2 \\
&=-\frac{\eta}{2}\,(\mu_\alpha-\varpi)^2\,\xi_\alpha\,dt
+\sqrt{\eta}\,(\mu_\alpha-\varpi)\,\xi_\alpha\,dW.
\end{align}
For $\alpha\neq\beta$, It\^o’s product rule gives
\begin{align}
\begin{aligned}
    d(\xi_\alpha\xi_\beta)
&=(d\xi_\alpha)\,\xi_\beta+\xi_\alpha\,(d\xi_\beta)+(d\xi_\alpha)(d\xi_\beta) \\
&=-\frac{\eta}{2}\![(\mu_\alpha-\varpi)^2+(\mu_\beta-\varpi)^2 \\
&\quad-2(\mu_\alpha-\varpi)(\mu_\beta-\varpi)]\xi_\alpha\xi_\beta\,dt \\
&\quad +\sqrt{\eta}\,(\mu_\alpha+\mu_\beta-2\varpi)\,\xi_\alpha\xi_\beta\,dW \\
&=-\frac{\eta}{2}\,(\mu_\alpha-\mu_\beta)^2\,\xi_\alpha\xi_\beta\,dt \\
&\quad+\sqrt{\eta}\,(\mu_\alpha+\mu_\beta-2\varpi)\,\xi_\alpha\xi_\beta\,dW.
\end{aligned}
\end{align}
Let $A$ denote the infinitesimal generator on $C^2$ test functions 
$f$ of the Markov process $p(t)$,
\begin{align}
A f\big(p(t)\big)
:=\lim_{h\to 0}\frac{1}{h}\Big(
\mathbb{E}[f(p(t+h))\,|\,\mathcal{F}_t]-f(p(t))\Big).
\end{align}
Taking conditional expectations yields
\begin{align}
A\big(\xi_\alpha\xi_\beta\big)
=-\frac{\eta}{2}\,(\mu_\alpha-\mu_\beta)^2\,\xi_\alpha\xi_\beta.
\end{align}
By definition of the inter–subspace spectral gap
\begin{align}
\Delta_{\alpha\beta}
:=\min\big\{|\lambda_j-\lambda_k|:\ 
\tilde\rho_j\le\Pi_\alpha,\ \tilde\rho_k\le\Pi_\beta\big\},
\end{align}
and since $\mu_\alpha(t)$ and $\mu_\beta(t)$ are convex combinations of
$\{\lambda_j\}_{j\in J_\alpha}$ and $\{\lambda_k\}_{k\in J_\beta}$, 
we have $|\mu_\alpha-\mu_\beta|\ge \Delta_{\alpha\beta}$ whenever 
$p_\alpha p_\beta>0$. Hence
\begin{align}
A\big(\xi_\alpha\xi_\beta\big)
\le -\frac{\eta}{2}\,\Delta_{\alpha\beta}^2\,\xi_\alpha\xi_\beta.
\end{align}
Summing over $\alpha<\beta$ gives
\begin{align}
\begin{aligned}
    A V(\rho)
&= \sum_{\alpha<\beta}A(\xi_\alpha\xi_\beta)\\
&\le -\,\frac{\eta}{2}\Big(\min_{\alpha\ne\beta}\Delta_{\alpha\beta}^2\Big)
\sum_{\alpha<\beta}\xi_\alpha\xi_\beta
= -\,r\,V(\rho),
\end{aligned}
\end{align}
with $r:=\tfrac{\eta}{2}\min_{\alpha\ne\beta}\Delta_{\alpha\beta}^2$. 
Therefore $V$ is a nonnegative supermartingale with exponential decay, and
\begin{align}
\mathbb{E}[V(\rho_t)]\le e^{-rt}V(\rho_0).
\end{align}

(iii) Since $V(\rho)\ge 0$ and $\mathbb{E}[V(\rho_t)]\to 0$, and
$V(\rho)=0$ if and only if some $p_\alpha=1$, it follows that
$V(\rho_t)\to 0$ almost surely; hence $\rho_t$ converges almost surely to
$\{\rho:\exists\alpha,\ p_\alpha(\rho)=1\}$. This proves (iii).

(iv) Because $p_\alpha(t)$ is a bounded martingale, we have
\[
\mathbb{E}[p_\alpha(\infty)]
=\mathbb{E}[p_\alpha(0)]
=p_\alpha(0)=\Tr(\rho(0)\Pi_\alpha).
\]
From (iii), $p_\alpha(\infty)\in\{0,1\}$ almost surely, so
\[
\mathbb{P}\big(\rho(\infty)\in\Pi_\alpha\big)
=\mathbb{E}[p_\alpha(\infty)]
=\Tr(\rho(0)\Pi_\alpha),
\]
which establishes (iv).
\end{proof}

Based on Lemma~\ref{lem:inv-subspace}, the infinite-horizon SMPC problem with cost 
function \eqref{eqn:SMPC-cost} remains tractable under the invariant–subspace 
assumptions. In particular, in the nondegenerate rank-one case with 
$\rho_f=\Pi_f$, the result reduces to Theorem~\ref{thm:SMPC_equivalence}, 
and this holds even when $H_0 \neq L$.

\begin{remark}
In the general case, our analysis requires the target state to be an eigenstate of the operator defining the invariant subspaces used in the SME reduction.
Practically, one may introduce a constant control bias \( \bar{u} \) such that
\[
\widetilde{H}_0 = H_0 + \bar{u} H_c,
\]
which ensures that the desired projector \( \Pi_{\mathrm{tar}} \) becomes an eigenprojector of \( \widetilde{H}_0 \).
If there exists a family of orthogonal projectors \( \{\Pi_\alpha\}_{\alpha=1}^m \) satisfying
\[
[\widetilde{H}_0, \Pi_\alpha] = 0, \qquad [L, \Pi_\alpha] = 0, \qquad \sum_{\alpha=1}^m \Pi_\alpha = I,
\]
then Assumption \ref{ass:inv-subspaces} holds with \( H_0 \) replaced by \( \widetilde{H}_0 \).
Consequently, Lemma \ref{lem:inv-subspace} and Theorem \ref{thm:SMPC_equivalence} remain valid under this constant bias, and the infinite-horizon cost still reduces to the one-step fidelity form.
\end{remark}

\section{Stability Analysis}
\label{sec:stability}

In this section, we analyze the stability properties of the proposed
stochastic model predictive control framework.
Building on Assumption~\ref{ass:interval_contraction_SME},
we show that the expected cost decreases monotonically
along the closed-loop trajectory and that the conditional state
converges to the desired target~$\Pi_f$.

\begin{assumption}[Contraction over a finite interval]
\label{ass:interval_contraction_SME}
Let $\mathcal U[t_a,t_b]$ denote the set of admissible open-loop controls $u(\cdot)$
on the finite interval $[t_a,t_b]$.
Consider the stochastic process $\rho_t$ governed by the SME~\eqref{eqn:SME} under
$u(\cdot)\in\mathcal U[t_a,t_b]$ with initial state $\rho_{t_a}$.
There exists a contraction factor $\gamma_{[t_a,t_b]}\in(0,1)$ such that, for every
$\rho_{t_a}$, there is a $u(\cdot)\in\mathcal U[t_a,t_b]$ satisfying
\begin{align}
2\!\left(1-\Tr\!\big(\rho_{t_b}\Pi_f\big)\right)
\;\le\;
\gamma_{[t_a,t_b]}\,
2\!\left(1-\Tr\!\big(\rho_{t_a}\Pi_f\big)\right)
\quad \text{a.s.}.
\label{eq:interval_contraction_pathwise}
\end{align}
\end{assumption}

Guided by Theorem~\ref{thm:SMPC_equivalence}, we design the controller using the averaged dynamics. These are obtained by taking the expectation over the measurement noise in the SME~\eqref{eqn:SME}, yielding the unconditional Lindblad dynamics
\begin{align}
\begin{aligned}
\frac{d}{ds}\,\rho(s)
&= \mathbb{E}_\omega\!\left[\frac{d}{ds}\,\rho_s^\omega\right] \\[2pt]
&= -\,i\!\left[\,H_0 + H_u(s),\,\rho(s)\,\right]
   + \kappa\,\mathcal D[L]\rho(s),
\end{aligned}
\label{eq:uncond-Lindblad-expect}
\end{align}
where $H_u(s)=u(s)H_c$. For convenience, given any $t_a\le t_b$, the associated stochastic and averaged
channels are defined by
\begin{align}
\rho_{t_b}^\omega  \;:&=\; \Phi_{[t_a,t_b]}(\rho_{t_a}),
\\
\bar\Phi_{[t_a,t_b]}(\rho_{t_a})
\;:&=\;\mathbb E_\omega\!\big[\Phi_{[t_a,t_b]}(\rho_{t_a})\big].
\label{eq:channels-def}
\end{align}

\begin{proposition}
\label{prop:expected_contraction}
Let $\{t_k\}$ be the sampling instants and define the predicted cost
\begin{align}
\label{eq:Jk_def}
J_k
&=\;
2\!\left(
1-\Tr\!\big(
\bar\Phi_{[t_k,\,t_k+\Delta t]}(\rho_{t_k})\Pi_f
\big)
\right),
\end{align}
where $\bar\Phi_{[t_a,t_b]}$ denotes the unconditional channel generated by
an admissible open-loop control
$u(\cdot)\in\mathcal U[t_a,t_b]$
satisfying Assumption~\ref{ass:interval_contraction_SME}.  
Then, for each $k$, there exists a contraction factor $\gamma_k\in(0,1)$ such that
\begin{align}
\label{eq:Jk_contr}
\mathbb{E}[J_{k+1}]
\;\le\;
\gamma_k\,\mathbb{E}[J_k].
\end{align}
Define the uniform bound
\begin{align}
\label{eq:gamma_bar_def}
\bar\gamma
\;:=\;
\sup_k \gamma_k,
\qquad
0<\bar\gamma<1.
\end{align}
Then $\mathbb{E}[J_k]\!\to\!0$ geometrically as $k\!\to\!\infty$,
and the conditional state $\rho_t$ converges to $\Pi_f$
in expectation and almost surely.
\end{proposition}

\begin{proof}
We distinguish two cases depending on the relation between the prediction window 
$t_k+\Delta t$ and the next sampling instant $t_{k+1}$.

\medskip
\noindent
\textbf{Case 1.} $t_{k+1} \geq t_k+\Delta t$.\\
By Assumption~\ref{ass:interval_contraction_SME},
for any stochastic trajectory
$\rho_{t_{k+1}}^\omega$
(generated from
$\Phi_{[t_k,\,t_{k+1}]}(\rho_{t_k})$),
there exists a control 
$u(\cdot)\!\in\!\mathcal U[t_{k+1},\,t_{k+1}+\Delta t]$
such that
\begin{align}
\label{eq:assumed_contr_case1}
\begin{aligned}
&\Tr\!\big(
\bar\Phi^{\omega}_{[t_{k+1},\,t_{k+1}+\Delta t]}
(\rho_{t_{k+1}}^\omega)\Pi_f
\big) \\
&\ge\;
\gamma_{[t_{k+1},\,t_{k+1}+\Delta t]}^\omega\,
\Tr(\rho_{t_{k+1}}^\omega\Pi_f).
\end{aligned}
\end{align}
Here, $0<
\gamma_{[t_{k+1},\,t_{k+1}+\Delta t]}^\omega
<1$ and $\bar\Phi^{\omega}$ indicates that the open-loop control
on $[t_{k+1},\,t_{k+1}+\Delta t]$
is designed based on the stochastic initial state $\rho_{t_{k+1}}^\omega$.
Taking expectation and bounding the pathwise factors by their supremum gives
\begin{align}
\label{eq:case1-expect}
\begin{aligned}
&\mathbb{E}_\omega\!\left[
\Tr\!\big(
\bar\Phi^{\omega}_{[t_{k+1},\,t_{k+1}+\Delta t]}
(\rho_{t_{k+1}}^\omega)\Pi_f
\big)
\right]\\
&\ge\;
\gamma^{\star}_{[t_{k+1},\,t_{k+1}+\Delta t]}\,
\Tr(\rho_{t_{k+1}}\Pi_f),
\end{aligned}
\end{align}
where $\gamma^{\star}_{[t_{k+1},\,t_{k+1}+\Delta t]}:=\sup_\omega
\gamma_{[t_{k+1},\,t_{k+1}+\Delta t]}^\omega
\in(0,1)$.
Using Lemma~\ref{lem:inv-subspace}, the probability is a martingale, so it follows
\begin{align}
\label{eq:martingale}
\Tr(\rho_{t_{k+1}}\Pi_f)
\;=\;
\Tr(\rho_{t_k+\Delta t}\Pi_f),
\qquad
(t_{k+1} \geq t_k+\Delta t).
\end{align}
Combining
\eqref{eq:case1-expect}
and
\eqref{eq:martingale},
\begin{align}
\begin{aligned}
&\mathbb{E}_\omega\!\left[
\Tr\!\big(
\bar\Phi^{\omega}_{[t_{k+1},\,t_{k+1}+\Delta t]}
(\rho_{t_{k+1}}^\omega)\Pi_f
\big)
\right]\\
&\ge\;
\gamma^{\star}_{[t_{k+1},\,t_{k+1}+\Delta t]}\,
\Tr(\rho_{t_k+\Delta t}\Pi_f).
\end{aligned}
\end{align}
Since 
\begin{align}
J_k
&=2\!\left(
1-\Tr\!\big(
\bar\Phi_{[t_k,\,t_k+\Delta t]}(\rho_{t_k})\Pi_f
\big)
\right),\\
J_{k+1}
&=2\!\left(
1-\Tr\!\big(
\bar\Phi_{[t_{k+1},\,t_{k+1}+\Delta t]}
(\rho_{t_{k+1}})\Pi_f
\big)
\right),
\end{align}
we conclude
\begin{align}
\label{eq:case1-J}
\mathbb{E}[J_{k+1}]
\;\le\;
\gamma^{\star}_{[t_{k+1},\,t_{k+1}+\Delta t]}\,
\mathbb{E}[J_k].
\end{align}

\medskip
\noindent
\textbf{Case 2.} $t_k+\Delta t > t_{k+1}$.\\
Based on Theorem~\ref{thm:SMPC_equivalence},
the control is designed from the unconditional Lindblad equation \eqref{eq:uncond-Lindblad-expect}.  
Consider the suboptimal controller $u^\star$
defined on $[t_k,\,t_k+\Delta t]$.
Although
$\mathbb{E}_\omega\!\left[\Tr\!\big(
\bar\Phi^{\omega}_{[t_{k+1},\,t_{k+1}+\Delta t]}
(\rho_{t_{k+1}}^\omega)\Pi_f
\big)\right]$
is taken over different initial states,
we can apply the same suboptimal control $u^\star$ 
defined on $[t_k,\,t_k+\Delta t]$.
The segment $[t_k,\,t_{k+1}]$
already follows this controller,
and on $[t_{k+1},\,t_{k+1}+\Delta t]$
we can still use the same $u^\star$.
Hence, for this fixed suboptimal process,
\begin{align}
\label{eq:factor-decomp-1}
\begin{aligned}
\bar\Phi^{\omega}_{[t_{k+1},\,t_{k+1}+\Delta t]}
&=\;
\bar\Phi^{\psi}_{[t_{k+1},\,t_k+\Delta t]}
\circ
\bar\Phi^{\omega}_{[t_{k}+\Delta t,t_{k+1}+\Delta t]},
\end{aligned}
\end{align}
where
$\bar\Phi^{\psi}_{[t_{k+1},\,t_k+\Delta t]}$
denotes the channel with the same suboptimal control designed for $\rho_k$.
Since the continued control is suboptimal
with respect to the new control optimized at $t_{k+1}$, we have
\begin{align}
\label{eq:case2-a}
\begin{aligned}
&\mathbb{E}_\omega\!\left[
\Tr\!\big(
\bar\Phi^{\omega}_{[t_{k+1},\,t_{k+1}+\Delta t]}
(\rho_{t_{k+1}}^\omega)\Pi_f
\big)
\right]\\
&=\mathbb{E}_\omega\!\left[
\Tr\!\big(
\Phi^{\omega}_{[t_{k+1},\,t_{k+1}+\Delta t]}
(\Phi_{[t_k,\,t_{k+1}]}(\rho_{t_{k}}))\Pi_f
\big)
\right]\\
&\ge\;
\gamma^{\star}_{[t_{k}+\Delta t,\,t_{k+1}+\Delta t]}\,
\Tr\!\big(
\bar\Phi^{\psi}_{[t_{k+1},\,t_k+\Delta t]}
(\bar\Phi_{[t_k,\,t_{k+1}]}(\rho_{t_{k}})\Pi_f
\big).
\end{aligned}
\end{align}
Here, $\gamma^{\star}_{[t_{k}+\Delta t,\,t_{k+1}+\Delta t]}:=\sup_\omega
\gamma_{[t_{k}+\Delta t,\,t_{k+1}+\Delta t]}^\omega
\in(0,1)$ and $\rho_{t_{k+1}}^\omega$ can be obtained from the channel~\eqref{eqn:SME} with the suboptimal control applied to the initial state~$\rho_{t_k}$. Therefore,
\begin{align}
\begin{aligned}
&\mathbb{E}_\omega\!\left[
\Tr\!\big(
\bar\Phi^{\omega}_{[t_{k+1},\,t_{k+1}+\Delta t]}
(\rho_{t_{k+1}}^\omega)\Pi_f
\big)
\right]\\
&\ge
\gamma^{\star}_{[t_{k}+\Delta t,\,t_{k+1}+\Delta t]}\,
\Tr\!\big(
\bar\Phi_{[t_k,\,t_k+\Delta t]}
(\rho_{t_k})\Pi_f
\big).
\end{aligned}
\end{align}
This inequality directly yields
\begin{align}
\label{eq:case2-J}
\mathbb{E}[J_{k+1}]
\;\le\;
\gamma^{\star}_{[t_{k}+\Delta t,\,t_{k+1}+\Delta t]}\,\mathbb{E}[J_k].
\end{align}

Combining \eqref{eq:case1-J} and \eqref{eq:case2-J} yields a sequence
\(\{\gamma_k\}\subset(0,1)\) with \(\gamma_k\le \bar\gamma<1\), where
\(\bar\gamma:=\sup_k \gamma_k\), such that
\begin{align}
\begin{aligned}
\mathbb{E}[J_{k+1}]
&\le \gamma_k\,\mathbb{E}[J_k]
\;\le\; \bar\gamma\,\mathbb{E}[J_k].
\end{aligned}
\end{align}
Iterating gives
\begin{align}
\mathbb{E}[J_k]\;\le\;\bar\gamma^{\,k}\,\mathbb{E}[J_0]\xrightarrow[k\to\infty]{}0.
\end{align}

Define $\mathcal W_k:=2\bigl(1-\Tr(\rho_{t_k}\Pi_f)\bigr)\in[0,2]$ and since 
the definition of $J_k$ satifies \eqref{eq:Jk_def}.
Then the following holds:
\begin{align}
\label{eq:bridge}
\mathbb{E}\!\big[\mathcal W_{k+1}\mid\mathcal F_{t_k}\big]
\;\le\; J_k \;\le\; \gamma'_k\,\mathcal W_k
\quad\text{a.s.},
\end{align}
where $\gamma'_k\in(0,1)$ is the contraction factor from
Assumption~\ref{ass:interval_contraction_SME}.
The right inequality follows directly from Assumption~\ref{ass:interval_contraction_SME}
applied on $[t_k,t_k{+}\Delta t]$.
For the left inequality, note that by linearity of trace and conditional expectation,
$$\mathbb{E}[\Tr(\rho_{t_k+\Delta t}\Pi_f)\mid\mathcal F_{t_k}]
=\Tr(\bar\Phi_{[t_k,\,t_k+\Delta t]}(\rho_{t_k})\Pi_f).$$
If $t_{k+1}\ge t_k+\Delta t$, Lemma~\ref{lem:inv-subspace} yields equality
$$\mathbb{E}[\mathcal W_{k+1}\mid\mathcal F_{t_k}]=J_k;$$ if $t_{k+1}<t_k+\Delta t$,
the same controller on $[t_k,t_k+\Delta t]$ and monotonicity over
$[t_{k+1},t_k{+}\Delta t]$ imply
$$\mathbb{E}[\mathcal W_{k+1}\mid\mathcal F_{t_k}]\le J_k.$$

From \eqref{eq:bridge} we have the drift inequality
\begin{align}
\mathbb{E}[\mathcal W_{k+1}\mid\mathcal F_{t_k}]
\;\le\; \gamma_k\,\mathcal W_k
\quad\text{a.s.},
\end{align}
hence $(\mathcal W_k)$ is a bounded nonnegative supermartingale with
\(\mathbb{E}[\mathcal W_{k+1}] \le \bar\gamma\,\mathbb{E}[\mathcal W_k]\),
where $\bar\gamma:=\sup_k\gamma_k<1$.
It follows that $\mathbb{E}[\mathcal W_k]\to 0$ and, by the supermartingale
convergence theorem and Fatou’s lemma~\cite{durrett2019probability},
$\mathcal W_k\to 0$ a.s.
Therefore,
\[
\Tr(\rho_{t_k}\Pi_f)\xrightarrow[k\to\infty]{\text{a.s.}}1;
\]
i.e., the conditional state converges to the target projector $\Pi_f$ almost surely.
\end{proof}

\begin{remark}[Control of Angular Momentum Systems]
Consider an angular momentum system with $H_0=J_z$ and control Hamiltonians
$H_c\in\{J_x,J_y\}$, the continuous-time evolution satisfies
Assumption~\ref{ass:interval_contraction_SME}.  
Since $J_x$ and $J_y$ couple all adjacent eigenspaces of $J_z$,
any nonstationary state $\rho\neq\Pi_f$, thus strictly increases the fidelity
$\Tr(\rho_t\Pi_f)$ over time, ensuring a contraction factor
$\gamma_{[t_a,t_b]}\!\in(0,1)$ for every finite interval $[t_a,t_b]$.
\end{remark}

\begin{remark}[Control of General Systems]
For general open quantum systems where $H_0$ and $L$ are not simultaneously
diagonalizable.  
In such cases, the contraction property depends critically on the structure of
the control Hamiltonian~$H_c$.  
If $H_c$ fails to couple the relevant eigenspaces of~$H_0$,
the resulting evolution may stagnate within invariant manifolds, 
and the fidelity can remain constant.  
Therefore, for general systems one must carefully select $H_c$
to ensure nontrivial coupling between the dominant eigenspaces of~$H_0$
associated with the measurement operator~$L$,
so that a contraction factor $\gamma_{[t_a,t_b]}\!<1$ exists for each finite interval.
\end{remark}

\begin{remark}[Numerical Limitation]
In numerical implementations, however, the SME is integrated with a finite step size~$dt$,
and the total number of discrete updates $N=(t_b-t_a)/dt$ is finite.  
Consequently, for certain initial conditions or limited iteration counts,
the computed trajectory may not yet exhibit a strict decrease in the cost
within~$[t_a,t_b]$, even though the continuous dynamics remain contractive.
This finite-iteration effect is purely numerical and does not violate
the theoretical contraction property.
\end{remark}

\section{Optimal Control via Pontryagin’s Maximum Principle}
\label{sec:PMP}

To characterize the optimal control for the problem stated in Theorem~\ref{thm:SMPC_equivalence}, we apply PMP to the equivalent formulation in \eqref{eq:SMPC-final}. The resulting conditions provide a set of necessary criteria for local optimality within this predictive framework.

To obtain the averaged state \( \rho(t+\Delta t; u) \) defined in \eqref{eqn:avg_def}, we take the expectation of the SME~\eqref{eqn:SME} with respect to all measurement noise realizations \( \omega \). Let \( \rho_s^\omega \) denote the quantum trajectory corresponding to a particular realization of the Wiener process \( W_t \). Taking the expectation yields the deterministic evolution \eqref{eq:uncond-Lindblad-expect}.

To derive necessary conditions for optimality, we associate a Hermitian costate matrix \( \lambda(s) \) with the averaged state \( \rho(s) \), and define the corresponding Hamiltonian functional as
\begin{align}
    \mathcal{H}(\rho, u, \lambda) = \Tr\left[ \lambda(s) \left( -i[H_0 + H_u(s), \rho(s)] + \mathcal{D}[L]\rho(s) \right) \right].
\end{align}

The necessary conditions for optimality are:
\begin{itemize}
    \item \textbf{State equation}: \eqref{eq:uncond-Lindblad-expect} with initial condition \( \rho(s) = \rho(t) \).
    \item \textbf{Costate equation}:
    \begin{align}
        \dot{\lambda}(s) = -\frac{\partial \mathcal{H}}{\partial \rho}(\rho(s), u(s), \lambda(s)),
    \end{align}
    with terminal condition
    \begin{align}
        \lambda(t+\Delta t) = -2\bar{\rho}_f.
    \end{align}
    \item \textbf{Optimality condition}:
    \begin{align}
        u^*(s) = \arg\min_{u \in \mathcal{U}} \mathcal{H}(\rho(s), u, \lambda(s)).
    \end{align}
\end{itemize}

To analyze the structure of the optimal control, suppose the control Hamiltonian has the form \( H_u(s) = u(s) H_1 \), with a fixed Hermitian operator \( H_1 \) and scalar control \( u(s) \in [u_{\min}, u_{\max}] \). In this case, the switching function is given by
\begin{align}
\label{eqn:switching_fun}
    \mathcal{S}(s) := -i\, \Tr\left( \lambda(s) [H_1, \rho(s)] \right).
\end{align}
Then, the optimal control takes the bang-bang form:
\begin{align}
    u^*(s) =
    \begin{cases}
        u_{\max}, & \mathcal{S}(s) < 0, \\
        u_{\min}, & \mathcal{S}(s) > 0, \\
        \text{singular}, & \mathcal{S}(s) = 0.
    \end{cases}
\end{align}

\begin{remark}
The switching function \eqref{eqn:switching_fun} is always real valued. This follows from the fact that both the state \(\rho(s)\) and the costate \(\lambda(s)\) are Hermitian, and \(H_1\) is also Hermitian. Since the commutator \([H_1, \rho(s)]\) is anti-Hermitian, the product \(\lambda(s)[H_1, \rho(s)]\) is also anti-Hermitian. Hence, its trace is purely imaginary, and the prefactor \(-i\) ensures that \(\mathcal{S}(s)\in \mathbb{R}\).
\end{remark}

This switching rule is obtained by minimizing the Hamiltonian, which is linear in the control input, and yields an explicit characterization of the optimal control based on the sign of \( \Phi(s) \). The corresponding costate dynamics satisfy the adjoint equation:
\begin{align}
\label{eqn:adjoint equation}
    \dot{\lambda}(s) = -i[H_0 + u(s) H_1, \lambda(s)] + \mathcal{D}^\dagger[L]\lambda(s),
\end{align}
where \( \mathcal{D}^\dagger[L] \) denotes the adjoint of the Lindblad superoperator.

The resulting PMP-based formulation, together with the structure of the switching function, facilitates the synthesis of optimal control inputs over short prediction horizons. This control strategy is applied in numerical examples.

\section{Numerical Examples}
\label{sec:numerical}

To validate the proposed eigenstate-reduced SMPC strategy, we conduct numerical simulations on an open quantum system whose evolution is governed by the SME~\eqref{eqn:SME}. The objective is to steer the system state toward a target eigenstate \( \bar{\rho}_f \) by applying short-horizon controls derived via PMP.



\subsection{Stochastic Dynamics via Kraus Representation}

To numerically simulate the conditional evolution driven by~\eqref{eqn:dy}, we adopt the Kraus operator formalism following~\cite{rouchon2014models}. In each time interval \( [t, t + dt] \), the measurement increment \( dY_t \) is sampled according to~\eqref{eqn:dy}, and the quantum state evolves as
\begin{align}
\label{eqn: kraus}
\rho_{t+dt} = \frac{ M_{dY_t}\rho_t M_{dY_t}^\dagger + (1 - \eta)L \rho_t L^\dagger \,dt }{ \text{Tr} \left( M_{dY_t}\rho_t M_{dY_t}^\dagger + (1 - \eta)L\rho_t L^\dagger \,dt \right)},
\end{align}
where \( M_{dY_t} := I - (iH + \frac{1}{2}L^\dagger L)dt + \sqrt{\eta}\, dY_t L \).

\subsection{Control Design and Execution}
To determine the optimal control input over a finite prediction horizon \( [t, t+\Delta t] \), we solve the deterministic evolution
\begin{align}
\label{eq:uncond-Lindblad-expect-sim}
\dot{\rho}(s) = -i[H_0 + u(s) H_1, \rho(s)] + \mathcal{D}[L]\rho(s),
\end{align}
starting from the current state \( \rho(t) \). The objective is to minimize the terminal cost \( J = 1 - \Tr[\rho(t+\Delta t)\bar{\rho}_f] \), and the control function \( u(s) \) is optimized accordingly.

We adopt a gradient descent approach to update the control input over the prediction window. The costate \( \lambda(s) \), associated with the PMP, evolves backward in time as \eqref{eqn:adjoint equation} with the terminal condition \( \lambda(t+\Delta t) = -\bar{\rho}_f \). The control update is guided by the gradient of the cost function with respect to the control, approximated by the switching function \eqref{eqn:switching_fun}.

Following~\cite{lin2020time}, this gradient-based method is capable of capturing both bang-bang and singular control segments, even in the presence of switching functions \( \phi(s) \) that vanish over finite time intervals. This enables flexible and efficient exploration of optimal protocols without requiring full knowledge of analytic solutions or prior assumptions on control structure.

The resulting control trajectory \( u(s) \) is then applied to the stochastic evolution~\eqref{eqn: kraus} over the interval \( [t, t+\Delta t] \). The entire process is repeated iteratively at each decision point using the updated state and measurement feedback.

\subsection{Simulation Results}
\label{subsec:numerical}
To evaluate the stochastic performance of the proposed control strategy, we simulate \(1000\) trajectories of the stochastic master equation~\eqref{eqn: kraus}. Each trajectory evolves over a total horizon of \( T = 20 \) with time step \( dt = 0.01 \), resulting in \( N = 2000 \) steps. At each decision point \( t_k = k \Delta t \), with \( \Delta t = 0.5 \), the control input is optimized over the finite horizon \( [t_k, t_{k+1}) \), which consists of \( H_{\mathrm{step}} = \Delta t / dt = 50 \) steps. The detection efficiency is set to \( \eta = 1 \).

The system begins from a mixed initial state 
\[
\rho_0 = \mathrm{diag}(0.3,\, 0.4,\, 0.3),
\]
and is driven by a control Hamiltonian \( H_1 = J_y \), where
\[
J_y = \frac{1}{\sqrt{2}}\begin{bmatrix}
0 & -i & 0 \\
i & 0 & -i \\
0 & i & 0
\end{bmatrix}.
\]
The free Hamiltonian is given by \( H_0 =J_z = \mathrm{diag}(1,\, 0,\,-1) \), and the dissipation operator is defined as \( L = J_z \), modeling dephasing noise in the eigenbasis of \( H_0 \). The control input \( u_t \in \mathbb{R}\) is constrained to lie within the interval \( [-5,\,5] \) throughout the evolution.

During each horizon, the control function \( u(t) \) is optimized using gradient descent based on the deterministic dynamics~\eqref{eq:uncond-Lindblad-expect-sim}, which simulates the expected evolution. The costate dynamics are used to compute the gradient of the fidelity-based cost with respect to \( u(t) \), and the optimized control is then applied sequentially within the stochastic dynamics~\eqref{eqn: kraus} over \( H_{\mathrm{step}} \) steps. This process is repeated until the final time is reached.

Figure \ref{fig: three_avg} presents the averaged fidelity \( \mathbb{E}[\Tr(\rho(t)\bar{\rho}_f)] \), computed across 1000 sample paths. The blue line represents the SMPC-controlled evolution, while the red dashed line corresponds to the uncontrolled case with \( u(t) = 0 \). The results confirm that the SMPC strategy significantly improves convergence toward the target state \( \bar{\rho}_f \), demonstrating the effectiveness of the control framework under measurement backaction and decoherence.

\begin{figure}[htbp]
\centering
\includegraphics[width=0.98\columnwidth]{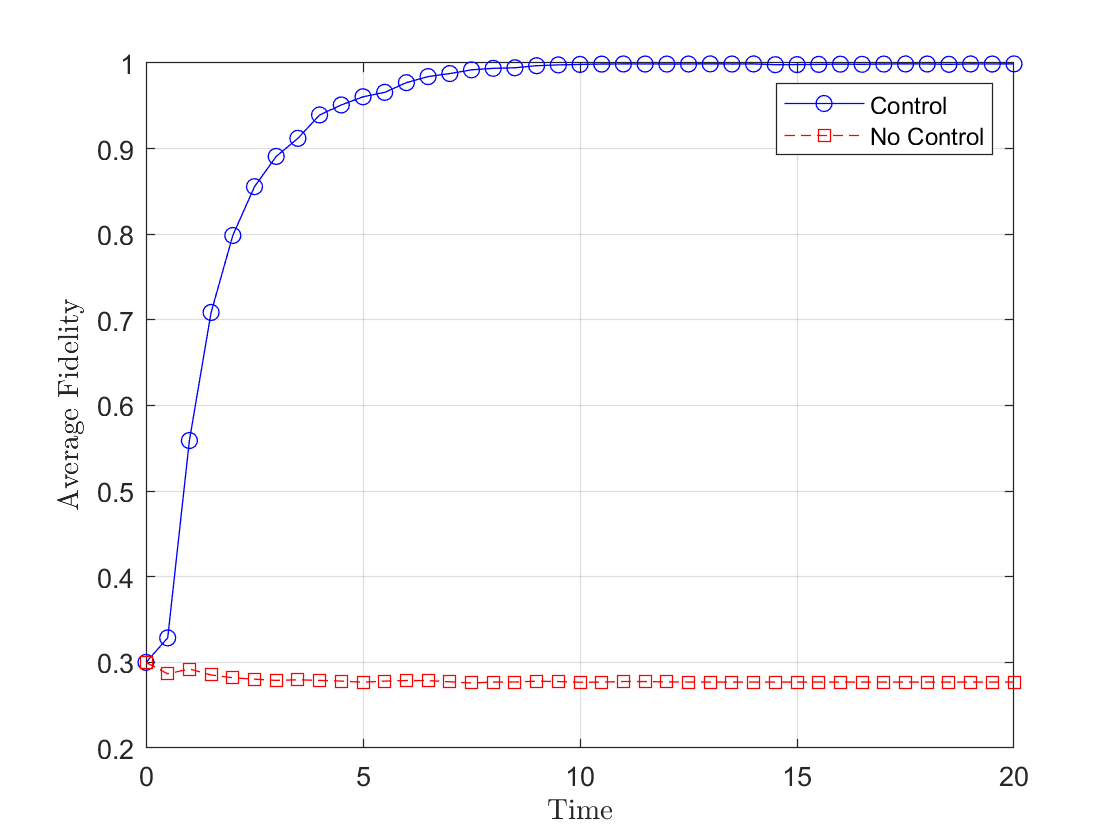}
\caption{Expected fidelity \( \mathbb{E}[\Tr(\rho(t)\bar{\rho}_f)] \) over time. The SMPC strategy (blue) shows significantly improved convergence toward the target state compared to the uncontrolled case (red).}
\label{fig: three_avg}
\end{figure}

\subsection{Comparison}
\label{sec:comparison}

We compare the proposed formulation with two references on the three-level system: the standard SMPC scheme described in the Appendix, and the exponential Lyapunov feedback of~\cite{liang2019exponential}.

The standard SMPC computes the expected cost by scenario sampling within each horizon, which increases online computation. The proposed method evaluates the same expected cost in a closed form by Theorem~\ref{thm:SMPC_equivalence}, so no sampling is required. Figure~\ref{fig:comparison} shows that the average fidelity has similar convergence, while the proposed controller requires less online computation. We also report the Lyapunov feedback of~\cite{liang2019exponential}. It enforces a Lyapunov decrease condition that guarantees exponential convergence and is included as a literature baseline; its design criterion differs from the horizon-based SMPC cost.

\begin{figure}[htbp]
\centering
\includegraphics[width=0.98\columnwidth]{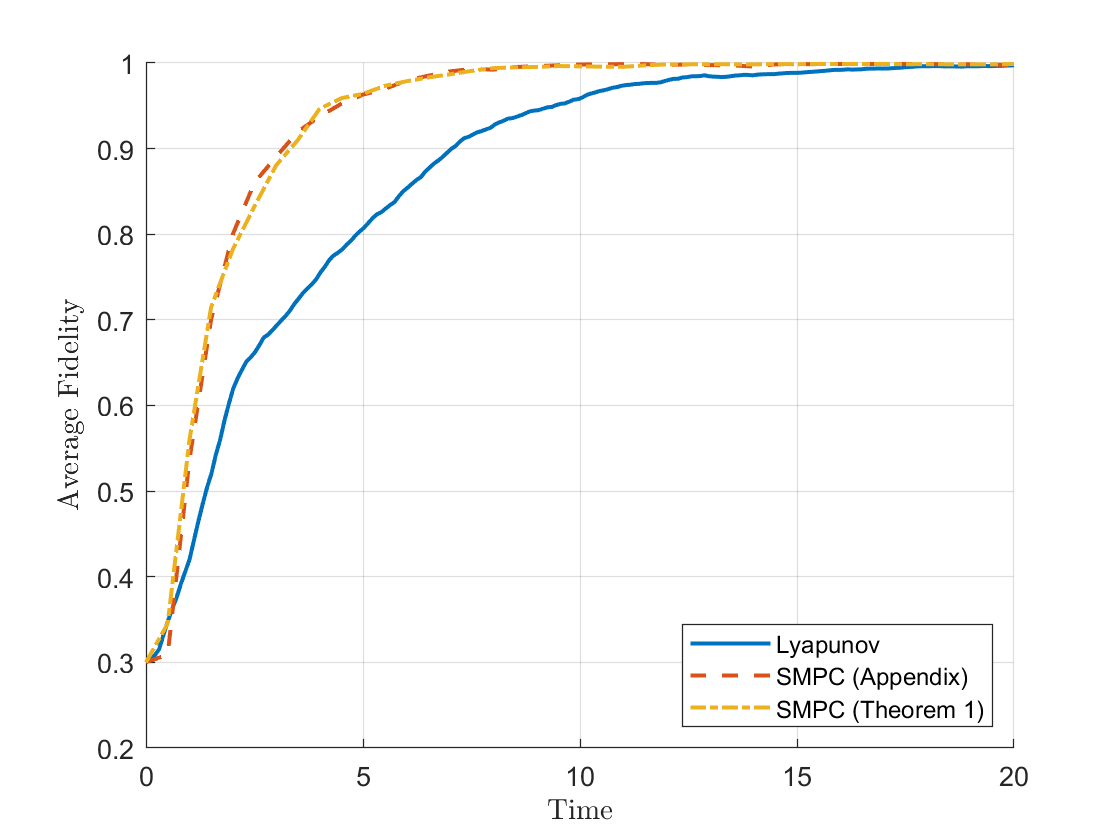}
\caption{Three-level system, average fidelity under Lyapunov feedback~\cite{liang2019exponential}, standard SMPC, and the proposed controller based on Theorem~\ref{thm:SMPC_equivalence}. The proposed method matches the standard SMPC and avoids scenario sampling.}
\label{fig:comparison}
\end{figure}

\subsection{Scalability with angular momentum $j$}
\label{sec:multilevel}

We repeat the procedure of Sec.~\ref{subsec:numerical} for angular-momentum systems with $j=1$ through $j=5$. 
Modeling and discretization follow Sec.~\ref{subsec:numerical} and \eqref{eq:uncond-Lindblad-expect-sim} together with the stochastic update~\eqref{eqn: kraus}. 
The initial state is $\rho_0 = \mathbb{I}_d/d$, where $\mathbb{I}_d$ denotes the $d\times d$ identity matrix. 
The target is the lowest $J_z$ eigenstate $\bar{\rho}_f = \ket{m=-j}\!\bra{m=-j}$, and the control bounds are unchanged.

The expected fidelity $\mathbb{E}[\Tr(\rho(t)\bar{\rho}_f)]$ is estimated from $1000$ trajectories for each $j$. 
Table~\ref{tab:fidelity_j} summarizes the results at the final time $T=150$, showing that the proposed controller consistently achieves high fidelity as the system dimension increases from $2j+1=3$ up to $11$.

As an illustration, Fig.~\ref{fig:j5} presents the trajectory of the averaged fidelity for the case $j=5$ (11 levels). 
This corresponds to the permutation-symmetric subspace of a 10-qubit system embedded in the $2^{10}$-dimensional Hilbert space. 
The figure demonstrates sustained convergence under continuous measurement in a relatively high-dimensional setting, confirming the scalability of the proposed approach.

\begin{figure}[htbp]
\centering
\includegraphics[width=0.98\columnwidth]{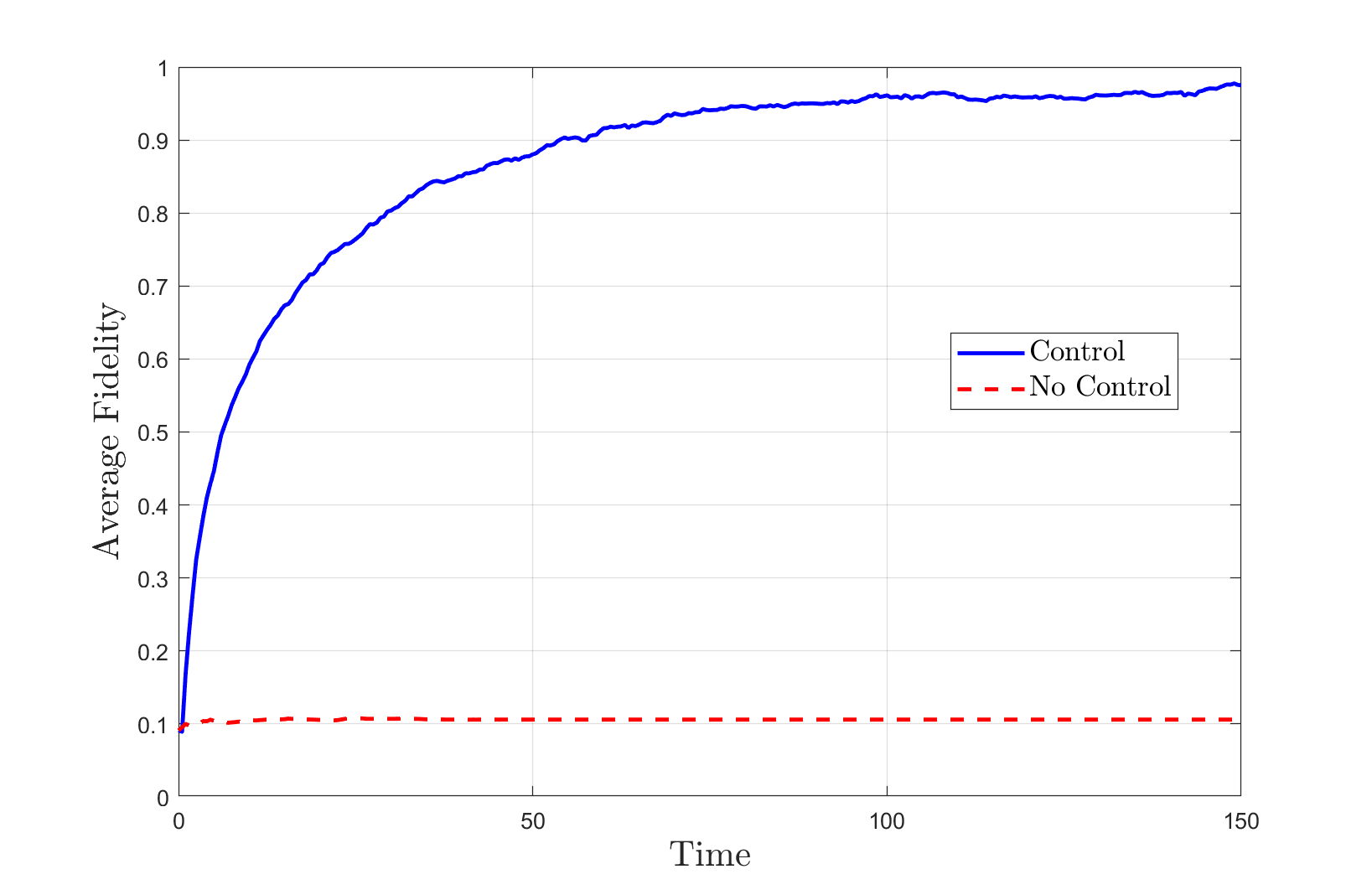}
\caption{Illustrative numerical example for $j=5$ (11 levels), corresponding to the permutation-symmetric subspace of a 10-qubit system. 
The process starts at $\rho_0=\mathbb{I}_d/d$; the averaged fidelity $\mathbb{E}[\Tr(\rho(t)\bar{\rho}_f)]$ over $1000$ trajectories shows sustained improvement under the proposed controller.}
\label{fig:j5}
\end{figure}

\begin{table}[htbp]
\centering
\caption{Expected fidelity $\mathbb{E}[\Tr(\rho(T)\bar{\rho}_f)]$ at $T=150$ for angular momentum indices $j=1,\dots,5$ (averaged over $1000$ trajectories).}
\label{tab:fidelity_j}
\begin{tabular}{c c}
\hline
Angular momentum $j$ & Fidelity at $T=150$ \\
\hline
1 & 0.9997 \\
2 & 0.9991 \\
3 & 0.9980 \\
4 & 0.9941 \\
5 & 0.9749 \\
\hline
\end{tabular}
\end{table}

\subsection{Ising Model (8-qubit system)}
\label{subsec:ising-model}

Previously, we focused on the angular momentum system and provided a comparison with the recent result in \cite{liang2019exponential}. 
Moreover, we demonstrated the scalability of our approach for angular momentum systems. 
Here, we consider a more general case based on Lemma~\ref{lem:inv-subspace}. 
This lemma can be extended to any system in which the target state is a common eigenstate of either $H_0$ or $L$. 
To illustrate this, we consider the following Ising-type model:
\begin{align}
    H_0 \;=\; 
    \sum_{(i,j)\in\mathcal{E}} J_{ij}\, Z_i Z_j 
    \;+\;
    \sum_{i=1}^n h_i Z_i,
    \label{eq:H0-Ising}
\end{align}
where $\mathcal{E}$ denotes the set of connected qubit pairs, $J_{ij}$ represents the Ising coupling strength, and $h_i$ is the local field on qubit $i$.

We then choose the product operators
\begin{equation}
  H_u \;=\; u\,\Big(\textstyle\bigotimes_{i=1}^{n} Y_i\Big), 
  \qquad 
  L \;=\; \textstyle\bigotimes_{i=1}^{n} Z_i ,
\end{equation}
so that $[H_0,\Pi_{\mathrm{tar}}]=0$ and $[L,\Pi_{\mathrm{tar}}]=0$, 
where $\Pi_{\mathrm{tar}}=\ket{00\cdots 0}\!\bra{00\cdots 0}$.
The control objective is to steer the system from the initial state $\ket{11\cdots 1}$ to the target eigenstate $\ket{00\cdots 0}$, i.e., 
$\rho_f=\ket{00\cdots 0}\!\bra{00\cdots 0}$.
Following our SMPC formulation, the design cost is selected as in~\eqref{eq:final-SMPC-cost} 
and evaluated under~\eqref{eq:uncond-Lindblad-expect} for the one-step prediction.

For execution, we simulate the stochastic dynamics~\eqref{eqn: kraus} for an eight-qubit system ($n=8$). 
The sampling period is set to $T_s=0.05$, and each sampling interval is numerically integrated with a substep size of $dt=0.0025$. 
The averaged fidelity is computed over $1000$ stochastic trajectories. 
Figure~\ref{fig:IsingAvg} shows that the averaged fidelity increases monotonically, confirming that the population concentrates within the invariant subspace associated with the target projector. 
This numerical result demonstrates that Theorem~\ref{thm:SMPC_equivalence} remains feasible under the relaxed condition established by Lemma~\ref{lem:inv-subspace}.

\begin{figure}[t]
  \centering
  \includegraphics[width=0.99\columnwidth]{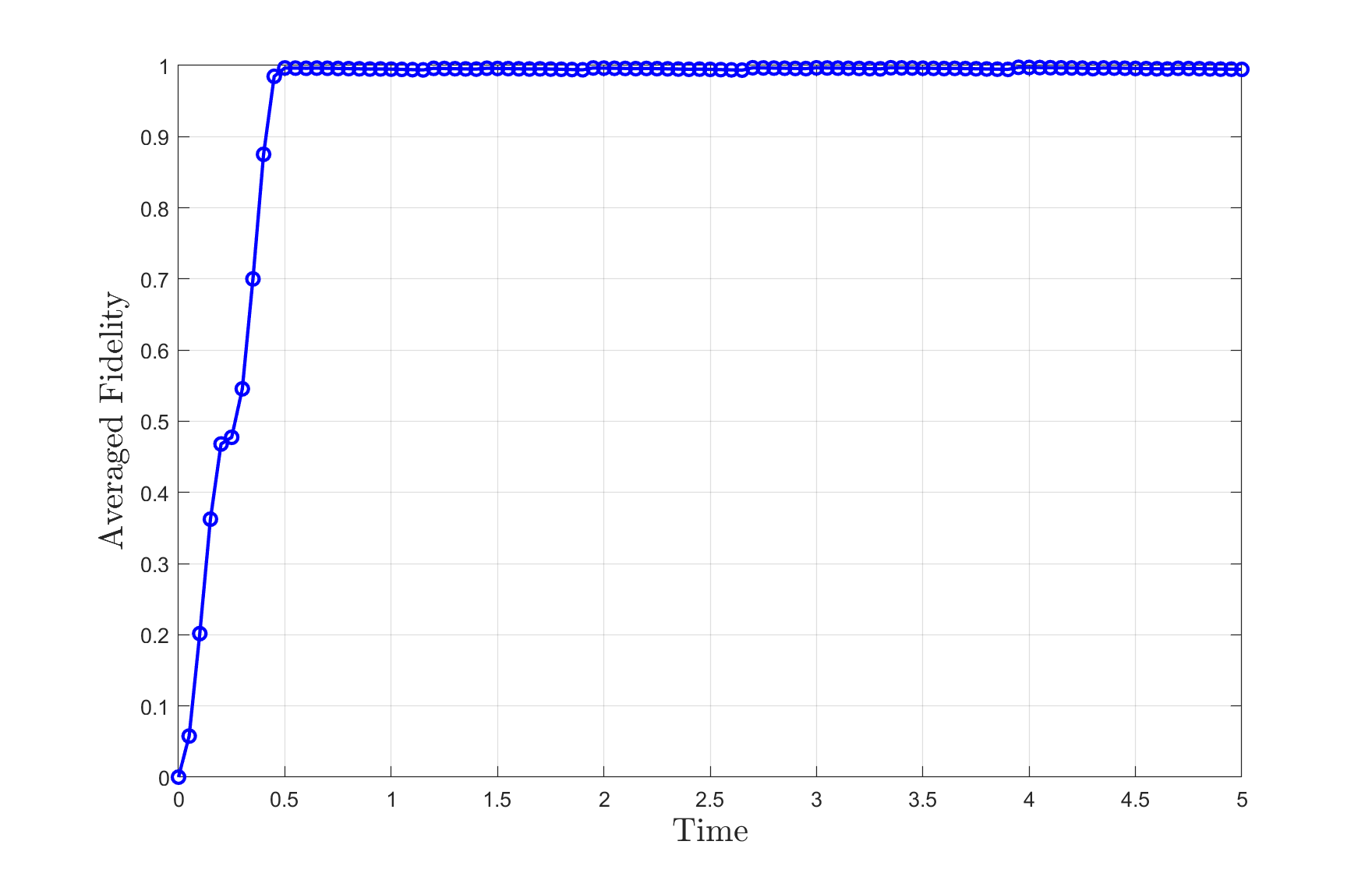}
  \caption{Averaged fidelity $\mathbb{E}[\Tr(\rho_t\rho_f)]$ versus time under control with $H_u=\bigotimes_i Y_i$ and $L=\bigotimes_i Z_i$. 
  The convergence to the target eigenstate verifies Theorem~\ref{thm:SMPC_equivalence} under the condition of Lemma~\ref{lem:inv-subspace}.}
  \label{fig:IsingAvg}
\end{figure}

\section{CONCLUSIONS AND FUTURE WORK}
In this work, we proposed a tractable formulation of infinite-horizon SMPC for quantum filtering systems by using the quantum state reduction property. This allows us to reformulate the original long-term cost into a short-horizon fidelity-based optimization problem. We further derived optimality conditions using PMP and identified a switching function structure for efficient control synthesis. The resulting approach avoids the complexity of Monte Carlo sampling and enables scalable implementation for multi-level quantum systems.

\appendix{}
\section{Discrete-Time Representation and SMPC Formulation}
As described in \cite{yang2013exploring}, we can represent the controlled quantum filtering equation \eqref{eqn:SME} in a real space. First, we define $n \times n$ matrices $X_1, \ldots, X_m$, where $m := n^2 - 1$ corresponds to the $n$-angular momentum system. These matrices satisfy the conditions: (i) $\text{tr}(X_l) = 0$, (ii) $\text{tr}(X_l X_j) = 2\delta_{lj}$, and (iii) $X_j^\dagger = X_j$. The anticommutator and commutator relations for these matrices are defined as follows:
\begin{equation}
\label{eqn: anticommutator}
\{X_l, X_j\} := X_l X_j + X_j X_l = \frac{4\delta_{lj}}{n} I_n + 2 \sum_{k=1}^{m} g_{ljk} X_k,
\end{equation}
\begin{equation}
\label{eqn: commutator}
-i[X_l, X_j] := -i(X_l X_j - X_j X_l) = 2 \sum_{k=1}^{m} f_{ljk} X_k,
\end{equation}
where $g_{ljk} := \frac{1}{4} \text{tr}(\{X_l, X_j\} X_k)$ and $f_{ljk} := \frac{1}{4i} \text{tr}([X_l, X_j] X_k)$. For instance, in a two-level quantum system, the Pauli matrices satisfy conditions \eqref{eqn: anticommutator} and \eqref{eqn: commutator}. In a three-level system, the Gell-Mann matrices satisfy these equations \cite{Kimura2003The}.

We can use these matrices to calculate real values from the density matrix as $x_l = \text{tr}(X_l \rho)$. These values define the coherent vector $x$:
\begin{equation}
\label{eqn: def coherent vector}
\mathbf{x} := (x_1, x_2, \ldots, x_m)^T \in B(\mathbb{R}^m) \subset \mathbb{R}^m,
\end{equation}
where $B(\mathbb{R}^m)$ is the coherent-vector space defined in $\mathbb{R}^m$. Taking the Euclidean norm, we find that $x$ lies within a ball of radius $\sqrt{\frac{2(n - 1)}{n}}$. The density operator can be written in the form:
\begin{equation}
\rho = \frac{I}{n} + \frac{1}{2} \sum_{l=1}^{m} x_l X_l.
\end{equation}
This formulation allows us to transform the quantum filtering equation into:
\begin{align}
\label{eqn: Bloch SME}
    d \mathbf{x} = \left(f(\mathbf{x}) + u(t) f_u(\mathbf{x})\right) dt + g(\mathbf{x}) dW, \quad x(0) = \mathbf{x}_0,
\end{align}
where the drift term $f(\mathbf{x}) = (\mathcal{L}_{H_0} + \mathcal{L}_D)\mathbf{x} + f_0$ and the control term $f_u (\mathbf{x}) = \mathcal{L}_{\mu}\mathbf{x}$ are $m \times 1$ column vectors. The stochastic term is given by $g(\mathbf{x}) = L_W(\mathbf{x})\mathbf{x} + L_{W_0}$. The superoperators $\mathcal{L}_{H_0}$, $\mathcal{L}_{\mu}$, and $\mathcal{L}_D$ are $m \times m$ matrices with elements given by:
\begin{small}
    \begin{align}
\begin{aligned}
     (\mathcal{L}_D)_{lr} =& -\delta_{lr} \frac{n}{2} \text{tr} \left( \sum_{j=1}^m L_j^\dagger L_j - \frac{1}{2} \sum_{j=1}^m \sum_{k=1}^{m} g_{lrk} \text{tr}(X_k L_j^\dagger L_j) \right) \\
&+ \frac{1}{2} \sum_{j=1}^{m} \text{tr}(X_l L_j X_r L_j^\dagger),
\end{aligned}
\end{align}
\end{small}
\begin{small}
\begin{align}
\label{eqn: LH0 LHmu}
    (\mathcal{L}_{H_0})_{lr} = -\sum_{k=1}^{m} f_{lrk} \text{tr}(X_k H_0), \quad (\mathcal{L}_\mu)_{lr} = -\sum_{k=1}^{m} f_{lrk} \text{tr}(X_k \mu),
\end{align}
\end{small}
and $(f_{0})_i=\frac{1}{n} \text{tr} \left(\sum_{j=1}^{m} [L_j, L_j^\dagger] X_i \right)$ is an $m \times 1$ column vector, where $i$ is the $i$-th element of the column vector $f_0$.
Note that \cite{yang2013exploring} contains a typo in \eqref{eqn: LH0 LHmu}, which can be verified in Section III of \cite{yang2013exploring}. Moreover, for the stochastic term $g(\mathbf{x})$, we calculate:
\begin{align}
\begin{aligned}
    (L_W(x))_{lr} = &-\frac{\delta_{lr}}{2} \sum_k \text{tr}(X_k (L + L^\dagger))x_k \\
    &+ \frac{1}{2} \sum_k g_{lrk} \text{tr}(X_k(L + L^\dagger)) \\
    &- \frac{i}{2} \sum_k f_{lrk} \text{tr}(X_k(L - L^\dagger))
\end{aligned}
\end{align}
and 
$(L_{W_0})_i = \frac{1}{n} \text{tr} ((L + L^\dagger)X_i)$ is an $m \times 1$ column vector, where $i$ is the $i$-th element of the column vector $L_{W_0}$.
For the stochastic part, we can define:
\begin{align}
    C_1 \triangleq 2\Re(\Gamma), \quad C_2 \triangleq 2\Im(\Gamma),
\end{align}
where $L =\Gamma_1 X_1+\Gamma_2 X_2+\ldots+\Gamma_m X_m$ and 
\begin{align}
\label{eqn: def Gamma}
    \Gamma:=(\Gamma_1,\Gamma_2,\ldots,\Gamma_m) \in \mathbb{C}^m .
\end{align}
Thus, we obtain:
\begin{align}
    (L_W(\mathbf{x}))_{lr} = -\delta_{lr} C_1 \mathbf{x} + \sum_k g_{lrk} (C_{1})_k + \sum_k f_{lrk} (C_{2})_k
\end{align}
and 
\begin{align}
    (L_{W_0}) = \frac{2}{n} C_1^T.
\end{align}

In the end, we model a discrete-time evolution corresponding to \eqref{eqn: kraus} and represent it within a coherent vector for \eqref{eqn: Bloch kraus} as 
\begin{align}
\label{eqn: Bloch kraus}
\begin{aligned}
     \mathbf{x}_{t+\delta t} = \mathcal{T} (\mathbf{x}_{t}, u_t,\delta W)
\end{aligned}
\end{align}
where $\delta W_t$ follows a Gaussian distribution, that is, $\delta W_t \sim N (0,\delta t) $ and $\mathcal{N}$ is the normalize factor. 
Note that, equation \eqref{eqn: Bloch kraus} is the same as \eqref{eqn: Bloch SME} up to order $\delta t$. To have a higher-order correction, either the Euler-Milstein method as in \cite{rouchon2014models,rouchon2015efficient}, or following the approach of \cite{guevara2020completely} that guarantees complete positivity of the conditional dynamics to order \((\delta t)^2\).

Therefore, in the following, we utilize \eqref{eqn: Bloch kraus} as the candidate equation to describe the discrete time quantum filtering. Also, we notice that $\delta t$ is still small to approximate the controlled quantum filtering equation \eqref{eqn:SME}. Thus, in the practical simulation, we assume a real discrete-time evolution such that we take $n$-steps for the small time $\delta t$ to be one-step. We can assume that the one-step time scale is $\Delta t = n \delta t$. Then, the one-step discrete-time evolution is
\begin{align}
\begin{aligned}
      \mathbf{x}_{t+\Delta t} &=\underbrace{\mathcal{T}(\mathcal{T}(\cdots\mathcal{T}}_{n \text{ times}}(\mathbf{x}_t, u_t, \delta W)\cdots), \delta W), \delta W)\\
      &=:\mathcal{U}_t (\mathbf{x}_t,u_t,\delta W),
\end{aligned}
\end{align}
where \( \mathcal{U}_t \) represents the effective one-step discrete-time transformation resulting from \(n\) applications of \(\mathcal{T}\) over the interval $\Delta t$.

Using this representation, we formulate the following SMPC problem:
\begin{align}
\label{eqn: SMPC}
    &V_N = \min_{\{u_i\}_{i=0}^{N-1}} J_N = \mathbb{E}\left\{\sum_{i=0}^{N-1} \ell (X_{i|k})\right\}, \nonumber \\
    &\text{s.t.} \quad
    \begin{cases}
        X_{0|k} = X_k, & \text{(initial condition)} \\
        X_{i+1|k} = \mathcal{U}_t(X_{i|k}, u_i, \delta W_i), & \text{(stochastic dynamics)} \\
        |u_{i|k}| \leq B, & \text{(control constraints)}
    \end{cases}
\end{align}
where \(J_N\) denotes the expected cost function corresponding to minimizing the deviation from the desired final state \(X_f\), with \(\ell(X) = \|X - X_f\|^2\). The function \(\mathcal{U}\) represents the stochastic system dynamics influenced by the control action \(u_i\) and the Wiener process \(\delta W_i\). 
In this context, we assume a single input signal for the quantum system which satisfies a control bound defined by \(B \in \mathbb{R}^+\).

To prove a stability for SMPC, we introduce the following assumption for the optimal value function \(V_N\):
\begin{assumption}
\label{ass: stochastic_Lyapunov}
There exists a continuous and positive definite function \(\sigma: \mathbb{R}^3 \to \mathbb{R}_{\geq 0}\) and a constant \(\overline{a} \in \mathbb{R}^+\) such that the value function \(V_N(X)\) satisfies:
\begin{align}
    V_N(X) \leq \overline{a} \, \sigma(X),
\end{align}
and the cost function \(\ell(X)\) is bounded below by:
\begin{align}
    \ell(X) \geq \sigma(X)
\end{align}
for all control inputs \(u\) satisfying the control constraints.
\end{assumption}

\begin{proposition}
\label{prop: stochastic_Lyapunov}
Suppose Assumption~\ref{ass: stochastic_Lyapunov} holds. Then, the system under the SMPC feedback law \eqref{eqn: SMPC} satisfies all constraints for all $k \in \mathbb{N}_{\geq 0}$, and
\begin{align}
    \mathbb{E} \{V_N(X_{k+1})\} - V_N(X_k) \leq -\ell (X_k) + \frac{\bar{a}^2}{N-1} \sigma(X_k).
\end{align}
Furthermore, there exists $N \geq \tilde{N}$ such that
\begin{align}
\label{eqn: exponentially_stable}
    \mathbb{E} \left\{V_N(X_{k+1})\right\} \leq (1-\beta) V_{N}(X_k),
\end{align}
for some $\beta \in (0, 1)$ in \( \mathbb{R} \).
\end{proposition}

The proof follows via similiar reasoning as in~\cite{Lorenzen2019Stochastic}. Once the system reaches the target state $X_f$, the dynamics satisfy $\mathcal{U}(X_f, 0, \delta W_i) = X_f$, implying stochastic invariance and exponential convergence. 

\begin{proof}
The expected difference in the optimal value function from time \(k\) to \(k + 1\) is given by:
\begin{align}
    \label{eqn: main_proof_equation}
    \begin{aligned}
        &\mathbb{E} \{V_N(X_{k+1})\} - V_N(X_k) \\
        &\leq \sum_{i=1}^{N-j} \mathbb{E} \{\ell(X_{i|k}^{*})\} + \mathbb{E} \{V_j(X_{N-j+1|k}^{*})\} \\
        &\quad - \sum_{i=0}^{N-1} \mathbb{E} \{\ell(X_{i|k}^{*})\} \\
        &= \sum_{i=1}^{N-j} \mathbb{E} \{\ell(X_{i|k}^{*})\} + \mathbb{E} \{V_j(X_{N-j+1|k}^{*})\} \\
        &\quad - \sum_{i=0}^{N-j} \mathbb{E} \{\ell(X_{i|k}^{*})\} - \sum_{i=N-j+1}^{N-1} \mathbb{E} \{\ell(X_{i|k}^{*})\} \\
        &\leq -\ell(X_k) + \mathbb{E} \{V_j(X_{N-j+1|k}^{*})\} \\
        &\leq -\ell(X_k) + \overline{a} \mathbb{E} \{\sigma(X_{N-j+1|k}^{*})\}.
    \end{aligned}
\end{align}
The first inequality in \eqref{eqn: main_proof_equation} holds for any \(j \leq N\) and is derived by using the feasible, though potentially suboptimal, control sequence \(u_{i|k+1} = u_{i+1|k}\) for \(i \leq N - j\). The last inequality follows from Assumption~\ref{ass: stochastic_Lyapunov}. 

By Assumption~\ref{ass: stochastic_Lyapunov}, the upper limit on \(V_N\) gives:
\begin{align}
    \sum_{i=1}^{N-1} \mathbb{E} \{\ell(X_{i|k}^{*})\} \leq \overline{a} \sigma(X_k),
\end{align}
which implies that there exists some \(i' \in \{1, \ldots, N-1\}\) such that:
\begin{align}
    \mathbb{E} \{\sigma(X_{i'|k}^{*})\} \leq \mathbb{E} \{\ell(X_{i'|k}^{*})\} \leq \frac{\overline{a}}{N-1} \sigma(X_k).
\end{align}
By choosing \(j = N - i' + 1\) and substituting into \eqref{eqn: main_proof_equation}, we obtain:
\begin{align}
    \mathbb{E} \{V_N(X_{k+1})\} - V_N(X_k) \leq -\ell(X_k) + \frac{\overline{a}^2}{N-1} \sigma(X_k).
\end{align}

For the second part of the proof, we consider the condition \(N \geq \overline{a}^2 + 1\). Under this condition, the inequality becomes:
\begin{align}
    \begin{aligned}
        -\ell(X_k) + \frac{\overline{a}^2}{N-1} \sigma(X_k) 
        &\leq \left(\frac{\overline{a}^2}{N-1} - 1\right) \sigma(X_k) \\
        &\leq \left(\frac{\overline{a}}{N-1} - \frac{1}{\overline{a}}\right) V_N(X_k).
    \end{aligned}
\end{align}
Thus, we have:
\begin{align}
    \begin{aligned}
        \mathbb{E} \{V_N(X_{k+1})\} 
        &\leq \frac{\overline{a}^2 + (\overline{a} - 1)(N-1)}{\overline{a}(N-1)} V_N(X_k) \\
        &:= (1 - \beta) V_N(X_k),
    \end{aligned}
\end{align}
where \(\beta = \frac{N - (\overline{a}^2 + 1)}{\overline{a} (N-1)} < 1\) for \(N > \overline{a}^2 + 1\). This concludes the proof.
\end{proof}

\bibliographystyle{IEEEtran}
\bibliography{reference}

\end{document}